\documentclass{article}

\usepackage{blindtext} 

\usepackage{amsmath,amssymb,amsfonts,amsthm,enumerate}
\usepackage[sc]{mathpazo} 
\usepackage[T1]{fontenc} 
\usepackage{microtype} 
\usepackage{fullpage}
\usepackage{graphicx}
\usepackage[english]{babel} 

\usepackage[hmarginratio=1:1,top=32mm,columnsep=20pt]{geometry} 
\usepackage[hang, small,labelfont=bf,up,textfont=it,up]{caption} 
\usepackage{booktabs} 

\usepackage{lettrine} 

\usepackage{enumitem} 
\setlist[itemize]{noitemsep} 

\usepackage{abstract} 
\newcommand{\sd}{\Sigma\Delta}
\newcommand{\R}{\mathbb{R}}
\usepackage{titlesec} 
\renewcommand\thesection{\Roman{section}} 
\renewcommand\thesubsection{\roman{subsection}} 
\titleformat{\section}[block]{\large\scshape\centering}{\thesection.}{1em}{} 
\titleformat{\subsection}[block]{\large}{\thesubsection.}{1em}{} 

\usepackage{fancyhdr} 

\usepackage{titling} 

\usepackage{hyperref} 
\newtheorem{Theor} {Theorem}
\newtheorem{rk}{Remark}
\newtheorem{cor}{Corollary} 
\newtheorem{defi}{Definition}

\newtheorem{lem}{Lemma}

\pretitle{\begin{center}\Huge\bfseries} 
\posttitle{\end{center}} 
\title{On one-stage recovery for $\Sigma \Delta$-quantized compressed sensing} 
\author{%
\textsc{Arman Arian and \"{O}zg\"{u}r Y\i lmaz} \\[1ex] 
\normalsize Department of Mathematics, \\ University of British Columbia \\ 
}
\date{} 
\renewcommand{\maketitlehookd}{%
\begin{abstract}

Compressed sensing (CS) is a signal acquisition paradigm to simultaneously acquire and reduce dimension of signals that admit sparse representations. When such a signal is acquired according to the principles of CS, the measurements still take on values in the continuum. In today's ``digital'' world, a subsequent quantization step, where these measurements are replaced with elements from a finite set is crucial.  We focus on one of the approaches that yield efficient quantizers for CS: $\Sigma \Delta$ quantization, followed by a one-stage tractable reconstruction method, which was developed in \cite{rongrong} with theoretical error guarantees in the case of sub-Gaussian matrices. We propose two alternative approaches that extend the results of \cite{rongrong} to a wider class of measurement matrices including (certain unitary transforms of) partial bounded orthonormal systems and deterministic constructions based on chirp sensing matrices. \\

Index Terms--compressed sensing, quantization, noise-shaping, $\Sigma \Delta$ quantization, one-stage reconstruction

\end{abstract}
}

\begin{document}

\maketitle


\section{Introduction}
\lettrine[nindent=0em,lines=3]{C}ompressed sensing (CS) has recently emerged as a revolutionary sampling theory. This new theory is based on the empirical observation that various important classes of signals, such as audio and images, admit (nearly) sparse approximations when expanded with respect to an appropriate basis or frame, such as a wavelet basis or a Gabor frame. CS theory shows that one can recover such signals from only a few  linear, non-adaptive measurements. As such, CS provides a dimension reduction paradigm. However,  in today’s digitally driven world, every sampling theory needs to be accompanied by a quantization theory.  Next, we discuss this aspect of CS. 

Formally, a signal is a vector $x$ in $\mathbb{R}^n$, where $n$ is potentially large. We say that $x$ is $k$-sparse if $\|x\|_0\le k$ where $\|x\|_0$ is the cardinality of the support of $x=[x_1,\dots,x_n]^T$ defined as $\text{supp}(x):=\{j:\ x_j\ne 0\}$. The set of all $k$-sparse signals in $\mathbb{R}^n$ is denoted by $\Sigma_k^n$. 

Suppose $x\in \Sigma_k^n$ or it is {\it compressible}, i.e., it can be well approximated in $\Sigma_k^n$ such that $\sigma_k(x):=\min_{v\in\Sigma_k^n} \|x-v\|_1$ is small. Compressed measurements of $x$ are linear, non-adaptive measurements given by $y=\Phi x + \eta$. Here $\Phi$ is an $m\times n$ CS measurement matrix with $m\ll n$ and $\eta$ is additive noise. Consequently, the ``compressed" measurement vector $y$ is still real valued, this time in $\mathbb{R}^m$, with $m \ll n$.  As mentioned earlier, in the classical signal processing paradigm, such an acquisition or {\it sampling} stage is followed by {\it quantization} where the sample values are mapped from the continuum to a finite set. While quantization was mostly omitted in the early CS literature, there has been several recent papers that address this problem. The approaches in the literature focus mostly on either ``memoryless scalar quantizers'' (MSQ) or ``noise-shaping quantizers''. 

\subsection{Memoryless scalar quantization for CS} Suppose that $x\in \R^n$ and $y\in \R^m$ are as above. An MSQ with alphabet $\mathcal{A}$ rounds off each entry of $y$ (independently) to the closest element of $\mathcal{A}$ \cite{powell, krahmer3, chou}.  A special case of MSQ is the 1-bit quantizers, where  each measurement is replaced by its sign \cite{boufounos, plan, plan2, rbaraniuk}, i.e., $\mathcal{A}=\{\pm 1\}$. 

One way to analyze the error associated with MSQ is by interpreting the quantization error as additive noise. Such an approach shows that one can obtain an approximation $\tilde{x}$ using, for example, Basis Pursuit Denoise \cite{donoho1, tao}. In that case, we get an approximation error bound $\|x-\tilde{x}\|$ that is proportional to the quantizer resolution, say $\delta$. This theoretical upper bound as well as the empirical performance --see \cite{gunturk}-- does not improve by increasing the number of measurements $m$.  On the other hand, it was observed  in \cite{gunturk} that in  a two-stage recovery method where the Penrose-Moore pseudo-inverse is used in the second stage (after support recovery), the error $\|x-\tilde{x}\|$ is empirically $\mathcal{O}(\frac{1}{\sqrt{m}})$. Motivated by this, \cite{kate} shows that $\|x-\tilde{x}\|$ is bounded by the sum of two terms: one that is independent on $m$ but unobservably small in any realistic setting, and another that is indeed $\mathcal{O}(\frac{1}{\sqrt{m}})$, at least for a wide class of sub-Gaussian matrices with high probability. Similarly, it was also shown in the 1-bit CS context in \cite{plan} that for a fixed level of sparsity, the error in approximation using a specific convex minimization program decays as $\mathcal{O}(\frac{1}{m^{1/5}})$ up to a logarithmic factor. 

While these improved results show some decay as a function of $m$, this decay is mild, suggesting that MSQ does not utilize extra measurements efficiently. This leads us to noise-shaping quantizers. 

\subsection{Noise-shaping quantizers for CS} Noise-shaping quantizers were originally introduced in the context of analogue-to-digital (A/D) conversion of bandlimited signals \cite{inose}. These A/D convertors, called $\Sigma \Delta$ quantizers, became popular \cite{schreier} as they can be implemented using low-accuracy circuit elements and still produce high-accuracy approximations by oversampling. For many classes of signals it is much easier to oversample  on circuitry compared to using high-accuracy circuit elements, for example scalar quantizers $Q_{\delta}$ with very small $\delta$. 

Motivated by their efficiency in exploiting redundancy, $\Sigma \Delta$ quantizers were considered in the context of frame expansions (which are inherently redundant). Indeed, they were shown to yield approximations that improve as the redundancy increases in the contexts of Gabor frames \cite{oy1, oy2}, finite frames in $\mathbb{R}^d$  with certain regularity assumptions \cite{bpy1,  bpy2, bodmann}, Gaussian random frames \cite{gunturk}, and sub-Gaussian random frames \cite{krahmer-yilmaz, kst}.

These results in frame theory were instrumental in early work that proposed $\Sigma \Delta$ quantization in the setting of CS. In a nutshell, suppose $x \in \Sigma_{k}^n$, $\Phi \in \mathbb{R}^{m \times n}$ be an appropriate CS measurement matrix, and $y=\Phi x$ be the noise free compressive measurements. Also, let $q$ be obtained by quantizing $y$ using an $r$th order $\Sigma \Delta$ scheme and let $D$ be the difference matrix as in \cite[Section 2.1]{rongrong}.  In \cite{gunturk} a {\it two-stage recovery algorithm} was proposed: first, the support set $T=\mbox{supp}(x)$ is recovered or estimated. Then, the reconstruction vector $\hat{x}$ is given by $\hat{x}_{\Sigma \Delta}=Fq$ with $F=(D^{-r} \Phi_T)^{\dagger} D^{-r}$, where $\Phi_T$ denotes the restriction of $\Phi$ to its columns indexed by $T$. 
While this two-stage reconstruction approach yields superior decay in approximation error as the number of measurements $m$ increases --see \cite{gunturk,krahmer-yilmaz} -- there are two major caveats: The two-stage approach is not robust with respect to additive noise, and it imposes size requirements on the smallest non-zero entry of the sparse signal. 

\section{One-stage recovery for $\sd$-quantized CS}  
As a remedy to the issues mentioned above, \cite{rongrong} proposed a {\it one-stage reconstruction method } which computes the approximation $\tilde{x}$ to $x$ by solving the convex optimization problem


\begin{align} (\hat{x},\hat{\nu})  := \arg \min_{(z,\nu)} \|z\|_1 & \mbox{ s.t. }   \|D^{-r} (\Phi z +\nu-q) \|_2 \leq C_r \delta \sqrt{m},\notag \\ 
& \mbox{ and } \| \nu \|_2 \leq \epsilon \sqrt{m}. \label{onestagesol}
\end{align} 
Fix, now, any $\ell$ that is sufficiently large so that $\ell$ measurements suffice to recover $x$ from $\Phi x$ in the non-quantized CS setting. Then the approximation $\hat{x}$ obtained as above satisfies 
\begin{equation} \| \hat{x} - x \|_2  \leq C \Big( (\frac{m}{\ell})^{-r+1/2} +\sqrt{ \frac{m}{\ell}} \epsilon \Big) \end{equation} 
where $c$, $C$ are constants that do not depend on $m$, $\ell$, $n$.

Indeed, this method solves the issues mentioned in the previous section when the CS measurements are obtained via sub-Gaussian matrices and certain Fourier matrices \cite{wang2}.  On the other hand, it is not known if this one-stage recovery method enjoys recovery guarantees when we use other important classes of measurement matrices, e.g., random restrictions of discrete Fourier transform matrices (DFT), or various classes of deterministic measurement matrices. 

\subsection{Generalizing to other measurement systems}\label{genA}

In order to generalize the results of \cite{rongrong} to other classes of random matrices and also certain deterministic matrices, we isolate one main property, which we call \textit{(P1)}, that the measurement matrices must satisfy for such a generalization. 
\smallskip

\noindent \textbf{Property (P1).} Suppose that $\Phi$ is an $m \times n$ unnormalized CS measurement matrix, with (expected) column norm of $\sqrt{m}$. We say that $\Phi$ satisfies the property (P1) of order $(k,\ell)$ if the RIP constant of $\frac{1}{\sqrt{\ell}} (\Phi)_{\ell}$---where $(\Phi)_{\ell}$ is the restriction of $\Phi$ to its first $\ell$ rows---satisfies $\delta_{2k} < 1/9$. 
\smallskip

Note that sub-Gaussian matrices, and random restrictions of the DFT matrix satisfy this property with high probability for appropriate choices of $k$ and $\ell$ (see Section \ref{twonovel}). 

Let $y=\Phi x +\eta$, and $\| \eta\|_{\infty} \leq \epsilon$. Set $H:=[C_r D^r \ \frac{\epsilon}{\delta}I]$. Here $C_r$ is a constant that can depend on the order $r$ and in the specific case of an $r$th order greedy $\Sigma \Delta$ quantizer, $C_r=1/2$ \cite{rongrong}. Next, let $H=U\Sigma V^T$ be the singular value decomposition of $H$. With this notation, the approach used in \cite{rongrong} is to show that $U^T \Phi$ satisfies \textit{(P1)}. It is well-known that sub-Gaussian matrices satisfy (P1) and this is leveraged in \cite{rongrong} to show that $U^T \Phi$ satisfies \textit{(P1)} as well. Yet, this implication is non-trivial and not necessarily true, for example, when $\Phi$ is a partial DFT matrix. 

Here, we propose two ways to circumvent this issue. Specifically, we will devise two novel approaches where it will be sufficient that $\Phi$ (instead of $U^T \Phi$) satisfies \textit{(P1)}. 

\section{Two novel approaches} \label{twonovel}

\subsection{ \label{modified2} Approach 1: Using a modified measurement matrix}  It can be shown (similar to the proof of Theorem 1 in \cite{rongrong}) that one-stage reconstruction following $\Sigma \Delta$ quantization can be performed if 
\begin{enumerate}
\item $\Phi$ satisfies \textit{(P1)}, and 
\item measurements are obtained using $U \Phi$ as opposed to $\Phi$. 
\end{enumerate}
In particular,  under this condition, the reconstruction error is as in \cite[Theorem 1]{rongrong}. Specifically, the following holds.

\begin{Theor} \label{recoveryp1} Suppose that $\Phi$ is an $m \times n$ CS matrix, $x \in \R^n$, and $k<\ell\le m$ is such that $\Phi$ satisfies (P1) of order $(k,\ell)$. Suppose the measurements of $x$ are given by $y=\tilde{\Phi} x$, where $\tilde{\Phi}=U \Phi$ with $U$ as above, and  quantized by an $r$th-order $\sd$ scheme. Then $\hat{x}$, obtained via \eqref{onestagesol} after replacing $\Phi$ with $\tilde{\Phi}$ satisfies 
\begin{equation} \label{lm6} 
\|x-\hat{x}\|_2 \leq C \Big( (\frac{m}{\ell})^{-r+1/2}\delta +\frac{\sigma_k(x)}{\sqrt{k}} + \sqrt{\frac{m}{\ell}}\epsilon\Big) \end{equation}
where $C$ is a constant that does not depend on $m,\ell,n$.  
\end{Theor}

\subsection*{Implications for bounded orthonormal systems:} The initial matrices used in CS were all non-structured random matrices such as sub-Gaussian matrices. Using them came with at least two important caveats, namely, multiplying non-structured matrices with vectors is a long process and also storing them is costly and difficult. For these reasons, an important class of random matrices in CS are considered choosing random rows of Fourier matrices. Since these random matrices are structured, they solve the issues mentioned above. Another reason for using these matrices is that in some applications such as MRI \cite{mri} or tomographic imaging \cite{tomographic} the devices are designed in a way that they measure the coefficients of signals in the transform domain.   Using these matrices was first suggested by Cand\`es et al. \cite{candes2006}  to recover sparse signals using few measurements. The number of measurements was later improved by Rudelson et al. \cite{rudelson}. Specifically, it is shown in \cite{rudelson} that for a normalized $n \times n$ discrete Fourier transform (DFT) matrix  $\mathcal{F}^{(n)}$ whose $(k,j)$th entry is given by \begin{equation} \label{fourierdef} \mathcal{F}^{(n)}_{k,j}= \frac{1}{\sqrt{n}}e^{\frac{2 \pi i (j-1)(k-1)}{n}} \end{equation}  If the number of measurements $m$ satisfies $m= \mathcal{O}(k \log^4 n)$ , then the submatrix $\Phi$ consisting of $m$ rows of $\mathcal{F}^{(n)}$ satisfies RIP condition with high probability. 



 In this paper, we use a generalization of Fourier matrices, called Bounded Orthonormal Systems (BOS), as defined in \cite{foucart}.  If $U$ is a discrete BOS, by choosing $m$ random rows of $\sqrt{n}U$, one can obtain the random matrix $A=\sqrt{n}R_T U$ where $R_T: \mathbb{C}^n \to \mathbb{C}^m$ is the random operator that samples $m$ rows of $U$.  According to the following theorem, after proper normalization, such matrix $A$ satisfies RIP with high probability if the number of measurements is large enough and thus it can be used as a CS measurement matrix. 

\begin{Theor}  \label{fourier} \cite{foucart2} Let $A \in \mathbb{C}^{m \times n}$ be the random sampling matrix associated with a BOS with constant $K \geq 1$. If for $\delta \in (0,1)$, $$m \geq C K^2 \delta^{-2} k \ln^4 (n)$$ (for a universal constant $C>0$), then with probability at least $1-n^{-\ln^3 n}$ the restricted isometry constant $\delta_k$ of $\frac{1}{\sqrt{m}} A$ satisfies $\delta_k \leq \delta$.
\end{Theor}

\begin{cor} \label{ufourier} For a $k$-sparse signal $x \in \mathbb{R}^n$, we can use a Fourier matrix $\mathcal{F}^{(n)}$, Discrete Cosine Transform matrix $\mathcal{C}^{(n)}$, or Discrete Sine Transform $\mathcal{S}^{(n)}$ and consider $m_0$ to be the smallest value (obtained by Theorem \ref{fourier}) for which the corresponding measurement matrix satisfies RIP with $\delta_{2k}<  1/9$ with high probability. Next, set $\ell:=m_0$, and choose $m \geq \ell$ rows of $\mathcal{F}^{(n)}$, $\mathcal{C}^{(n)}$, $\mathcal{S}^{(n)}$ randomly and denote them by $\mathcal{F}^{(m,n)}$, $\mathcal{C}^{(m,n)}$, and $\mathcal{S}^{(m,n)}$ respectively. Then, measure $x$ using $U \mathcal{F}^{(m,n)}$ , $U \mathcal{C}^{(m,n)}$, or $U \mathcal{S}^{(m,n)}$. Let  $\hat{x}$ be the solution to $\eqref{onestagesol}$ with $\Phi$ replaced by one of the matrices mentioned here.   Then, the error in approximation using one-stage $\Sigma \Delta$ quantization satisfies \eqref{lm6} as we increase the number of measurements $m$. \end{cor} 

\begin{rk} Here, we show that computing the signal with $U \mathcal{F}^{(m,n)}$, $U \mathcal{C}^{(m,n)}$, or $U  \mathcal{S}^{(m,n)}$ is a fast process at least when $r=1$. First, note that an explicit formula for entries of $U$ is given in the case of $r=1$ in \cite{krahmer} : $$\begin{aligned} U_{k \ell} & =\sqrt{\frac{2}{n+1/2}} \cos \Big( \frac{2(k-1/2)(n-\ell+1/2) \pi}{2n+1} \Big)\\ & =  \sqrt{\frac{2}{n+1/2}} \cos \Big( \frac{(2k-1)(2n-2\ell+1) \pi/2}{2n+1} \Big) \\ & = \sqrt{\frac{2}{n+1/2}} \cos \Big( (2k-1) \pi/2-\frac{(2k-1)\ell \pi}{2n+1} \Big) \\& =  \sqrt{\frac{2}{n+1/2}} (-1)^{k+1} \sin \Big( \frac{(2k-1) \ell \pi}{2n+1} \Big) \end{aligned} $$  On the other hand, Discrete Sine Transform (DST) of type III is given by \cite{nikara}  \begin{equation} \label{dstdef} \mathcal{S}^{(n)}_{k \ell}= \sqrt{\frac{2}{n}} \sin \Big( \frac{ (2k-1) \ell \pi}{2n} \Big) \end{equation} Therefore, we can obtain entries of $U$ using a submatrix of $\mathcal{S}^{(2n+1)}$. The reason is that we can write $(k,2 \ell)th$ element of $\mathcal{S}^{(2n+1)}$ as $$\mathcal{S}^{(2n+1)}_{k,2 \ell}=  \sqrt{\frac{2}{2n+1}} \sin \Big( \frac{ 2(2k-1) \ell \pi}{4n+2} \Big)=  \sqrt{\frac{2}{2n+1}} \sin \Big( \frac{ (2k-1) \ell \pi}{2n+1} \Big)$$ which is same as $(k,\ell)th$ entry of matrix $U$ in absolute value up to a constant. We will also use the expression above for the entries of $U$ in order to to show that evaluating $Uy$ for a vector $y$ is fast. See Remark \ref{aboutu}.    \end{rk} 

\begin{rk} While the singular value decomposition of $D$ can be computed explicitly, to our knowledge an explicit formula for singular value decomposition of  $D^r$ with $r \geq 2$ is not known. Note, however, that one can estimate the singular values of $D^r$ using \textit{Weyl}'s inequalities \cite{gunturk, krahmer}. 
\end{rk}

\begin{rk} \label{issues} Alternatively, one could apply $U$ after collecting the measurements using $\mathcal{F}^{(m,n)}$, $\mathcal{C}^{(m,n)}$, or $\mathcal{S}^{(m,n)}$. Of course, this would require that we keep all $m$ analogue measurements in memory, at least until we apply $U$ still in analogue domains which is not practically feasible in applications when $m$ is large. We will propose a remedy in Section \ref{digitalbuffer}. 
\end{rk}  

\subsection*{Numerical experiments}
\label{nmodified}

In order to verify the results given in Theorem \ref{recoveryp1}, and in particular, given in Corollary \ref{ufourier}, we perform a numerical experiment. In this experiment, we fix the ambient dimension of signals to $n=200$,  the sparsity level to  $k=5$, and the quantization step to $\delta=0.1$. We consider the $m \times 200$ matrix $U \mathcal{F}^{m,200}$ as suggested by Corollary \ref{ufourier} with $m \in \{20, 30, 40, 50, 60, 70\}$ as the measurement matrix. For each value of $m$, we consider 20 signals in $\mathbb{R}^{200}$, random support $T \subseteq \{1,2,...,200\}$, and with non-zero entries chosen from normal Gaussian distribution. For each of these signals, we find the measurement vector, and subsequently quantize it using \textit{first} or \textit{second} order $\Sigma\Delta$ quantization. Next, we find $\hat{x}$, the solution to \eqref{onestagesol}, and we find the error in approximation. We take an average for the error for all 20 signals and move to the next value of $m$. The results are plotted in Figure \ref{modifiedfigure} in log-log scale. As we observe in this Figure, the error bounds decays as predicted in \eqref{lm6}.

\begin{figure}
    
\begin{center}
    \includegraphics[width=0.35\textwidth]{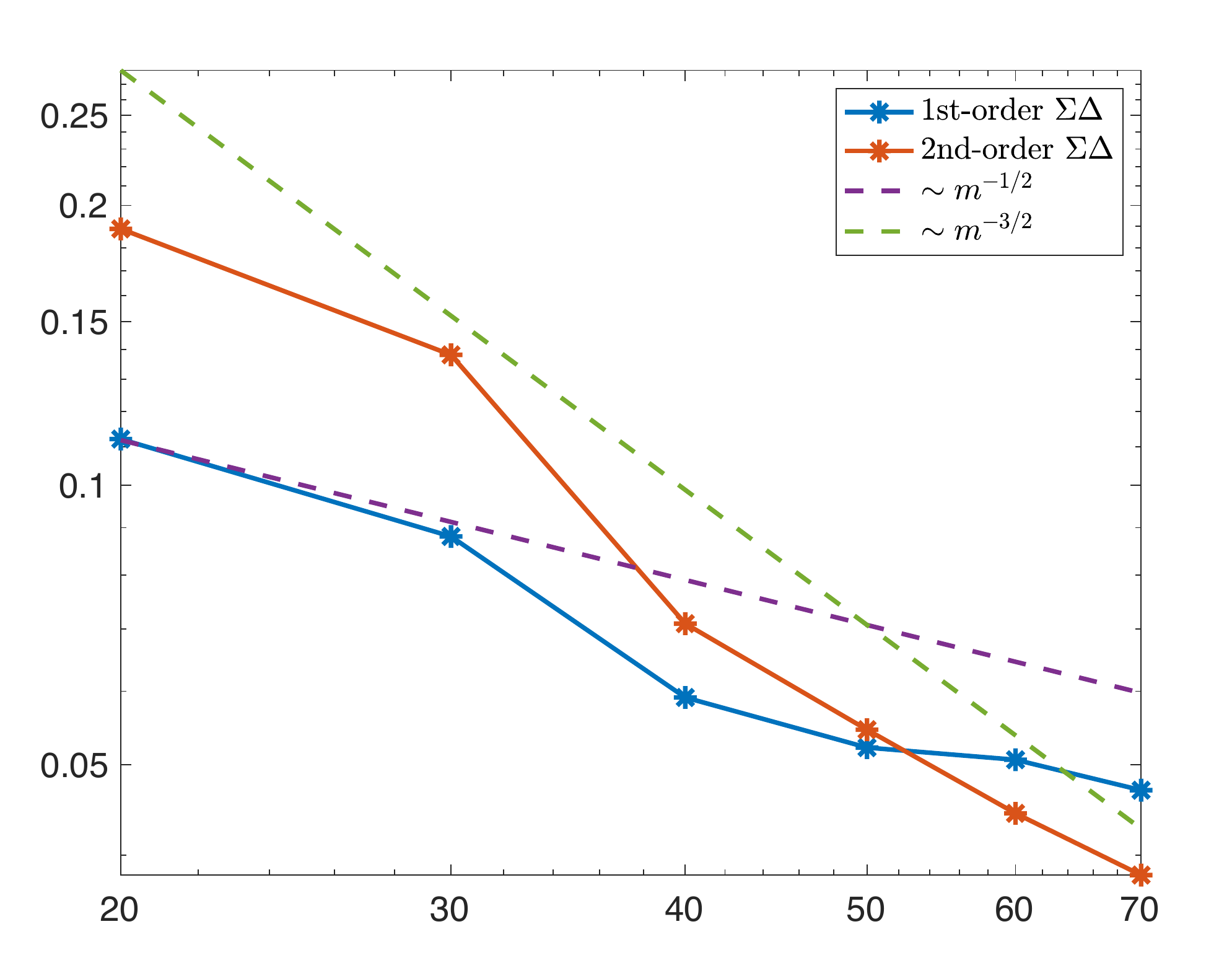}
\end{center}
\caption{  \label{modifiedfigure} Error in approximation  using  first order and second order $\Sigma \Delta$ quantization with one-stage reconstruction scheme and with a ``modified" random partial Fourier matrix  for 10-sparse signals and the comparison with the graphs of $f(m)=\frac{C}{m^{1/2}}$ and $g(m)=\frac{D}{ m^{3/2}}$  in log-log scale.  }

\end{figure}

\subsection{Approach 2: Using a digital buffer}
\label{digitalbuffer}

\begin{figure}
\begin{center}
\includegraphics[scale=0.35]{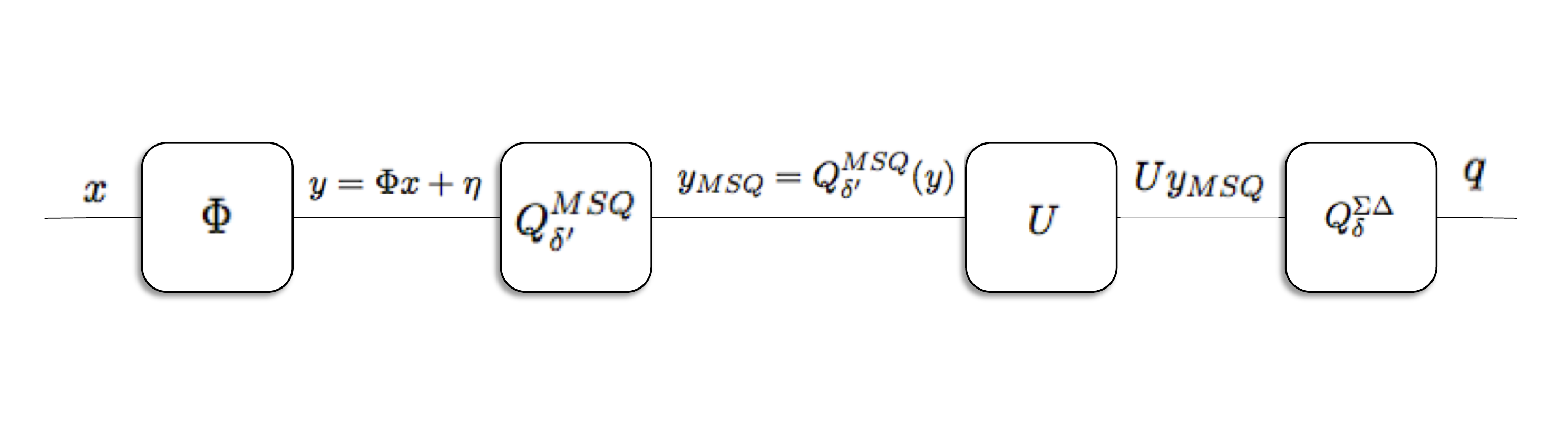}
\caption{Quantizing the signal $x$ by first using MSQ with a very small step size $\delta'$, then applying U which is a fast transform followed by a $\Sigma \Delta$ quantization scheme.}
\end{center}
\end{figure}
Aside from the issues raised in Remark \ref{issues}, the above approach  is not ideal also because the measurement matrix $U \Phi$ (specifically $U$) depends on $m$. This means that we must use a different measurement matrix if we wish to increase the number of measurements $m$, i.e., we can not ``reuse'' the measurements already collected. This problem would be resolved if we could modify the scheme so that

\begin{itemize}

\item We first collect $y=\Phi x$ and quantize $y$; 

\item We then use $U$ (or any other matrix that admits a fast implementation) on the quantized measurements, which are now in the digital domain. 

\end{itemize}

\noindent To that end, we propose the following scheme. 

\begin{enumerate}

\item Given a standard CS measurement matrix $\Phi$, we collect the compressed measurements $y=\Phi x+ \eta$, where $\eta$, as before, denotes the noise such that $\|\eta\|_{\infty} \leq \epsilon$. 

\item We fix a small $\delta'$ (much smaller than the desired final accuracy) and quantize $y$ using an MSQ with step size $\delta'$ resulting in $y_{MSQ}$. This is a high bit-budget representation of $y$ and will be discarded after the next stages so, it is just kept in a buffer (with sufficiently large memory). 

\item We compute $Uy_{MSQ} $, which finely approximates $Uy=U \Phi x$ as $U$ is an isometry. 

\item We use a $\Sigma \Delta$ quantizer (of appropriate order $r$ that matches the matrix $U$ in step (3)) with step size $\delta$ to quantize $Uy_{MSQ}$. This will be the digital representation of $x$ that we will keep. 

\end{enumerate}

\noindent Finally, we will reconstruct an approximation to $x$ by means of convex optimization problem similar to \eqref{onestagesol} given by 

\begin{equation}  \label{onestagesol2} \begin{aligned} (\hat{x},\hat{\nu})  := \arg \min_{(z,\nu)} \|z\|_1 & \mbox{ \ \ \  subject to \ \ \ }  \|D^{-r} (U\Phi z +\nu-q) \|_2 \leq C_r \delta \sqrt{m} \\ & \mbox{ \ \ \ \ and \ \ \ \ \  \ \ \ \ \ } \| \nu \|_2 \leq \delta'' \sqrt{m} \end{aligned} \end{equation}

\noindent with $\delta''$ defined as $\delta''=\epsilon+\delta'/2$.\\

Note that this method will be successful provided $\delta'$ in step (2) is sufficiently small to match the quantization error corresponding to the $\Sigma \Delta$ quantization of step (4). Thus, we will have to ensure that $m \leq m_{max}$ where $\delta'$ will be chosen depending on $m_{max}$ (or vice versa). Collecting all these, we have the following theorem, which we will prove after stating few remarks.  

\begin{Theor} \label{withmsq} Let $x \in \mathbb{R}^n$, $\Phi$ be a CS measurement matrix, $k$ and $\ell$ be such that $\Phi$ satisfies (P1), suppose that $q$ is obtained from $x$ following the scheme suggested above where \\

\begin{itemize}

\item $U$ is tailored to a $\Sigma \Delta$ quantizer of order $r$ (as described in Section \ref{modified2}).

\item $\delta':=\frac{\delta(3\pi r)^r}{m_{max}^r}$

\item $\|\eta\|_{\infty} \leq \epsilon $ 

\end{itemize}

\noindent If $\hat{x}$ is obtained via \eqref{onestagesol2}, the approximation error satisfies 

 \begin{equation} \label{lm66} \|x-\hat{x}\|_2 \leq 2 C_4 \delta (3 \pi r)^r (\frac{m}{\ell})^{-r+1/2} +2C_4 \sqrt{\frac{m}{\ell}} \epsilon+\frac{2C_5}{\sqrt{k}} \sigma_k(x) \end{equation}

\noindent for $\ell \leq m \leq m_{max}$.  Here, $C_4$ and $C_5$ depend only on the RIP constants of $\Phi$. 

\end{Theor}

\begin{rk} Since in practice, the original measurements in CS are physical quantities (such as currents), the MSQ step mentioned in Theorem \ref{withmsq} was performed in order to assign numbers to the measurements which enables us to store the measurements in the processor and multiply with $U$ later. 
\end{rk}

\begin{rk} \label{aboutu} In step (3) above, we need to compute $Uy_{MSQ}$. Here, we show that this computation can be done fast. To that end, we use the fact that for $\mathcal{S}^{(m)}$ as defined in \eqref{dstdef}, computing $\mathcal{S}^{(m)} y$ is fast for any positive integer $m$, and any vector $y$. Let $y=(y_1,y_2,...,y_m)$, then $$ \begin{aligned} & (Uy)_1=U_{11}y_1+U_{12}y_2+....+U_{1m}y_m \\ 
& (Uy)_2=U_{21}y_1+U_{22}y_2+...+U_{2m} y_m \\ & \ \ \ \ \ \  \vdots \\ & (Uy)_m=U_{m1}y_1+U_{m2}y_2...+U_{mm}y_m \end{aligned} $$
Thus,

$$ \begin{aligned} & (Uy)_1=  \sqrt{\frac{2}{m+1/2}} \Big(     \sin (\frac{\pi}{2m+1})  y_1+\sin (\frac{2 \pi}{2m+1})y_2+....+\sin (\frac{m \pi}{2m+1}) y_m \Big) \\ 
& (Uy)_2= -\sqrt{\frac{2}{m+1/2}} \Big( \sin (\frac{3 \pi}{2m+1})y_1 +\sin(\frac{6 \pi}{2m+1})y_2 +...+\sin(\frac{3m \pi}{2m+1})     y_m \Big) \\ & \ \ \ \ \ \  \vdots \\ & (Uy)_m=  (-1)^{m+1} \sqrt{\frac{2}{m+1/2}} \Big( \sin (\frac{ (2m-1)\pi}{2m+1})y_1 +...+\sin(\frac{m(2m-1) \pi}{2m+1})     y_m \Big)   \end{aligned} $$

\noindent  Thus, we can write the above equations in the following form. $$ \begin{aligned} & (Uy)_1=  \sqrt{2}\mathcal{S}^{(2m+1)}_{12}y_1+\sqrt{2}\mathcal{S}^{(2m+1)}_{14}y_2+...+\sqrt{2}\mathcal{S}^{(2m+1)}_{1,2m}y_m \\ 
& (Uy)_2=  -\sqrt{2}\mathcal{S}^{(2m+1)}_{22}y_1-\sqrt{2}\mathcal{S}^{(2m+1)}_{24}y_2-...-\sqrt{2}\mathcal{S}^{(2m+1)}_{2,2m}y_m     \\ & \ \ \ \ \ \  \vdots \\ & (Uy)_m=  (-1)^{m+1} \sqrt{2}\mathcal{S}^{(2m+1)}_{m2}y_1+(-1)^{m+1} \sqrt{2}\mathcal{S}^{(2m+1)}_{m4}y_2+...+(-1)^{m+1} \sqrt{2}\mathcal{S}^{(2m+1)}_{m,2m}y_m                \end{aligned} $$

\noindent Therefore, we have $$(Uy)_j= (-1)^{j+1} \sqrt{2} ( \mathcal{S}^{(2m+1)} \tilde{y})_j$$ for $j=1,2,...,m$, and where $\tilde{y} \in \mathbb{R}^{2m+1}$ is a vector whose odd entries are zero, and whose ($2k$)th entry ($k=1,2,...,m$) is defined as $y_k$. Accordingly, computing $Uy$ is a fast process.  

\end{rk}

To prove Theorem \ref{withmsq}, we use the following instrumental Lemma.


\begin{lem} \cite{foucartrobust} \label{hatx} Let $f, g \in \mathbb{C}^n$, and $\Phi \in \mathbb{C}^{m,n}$. Suppose that $\Phi$ is RIP with constant $\delta_{2k} < 1/9$. Then for any $ 1 \leq p \leq 2$, we have $$\|f-g\|_p \leq C_4 k^{1/p-1/2} \| \Phi (f-g) \|_2+\frac{C_5}{k^{1-1/p}} (\|f\|_1-\|g\|_1+2 \sigma_k(g)_{1}) $$ where $C_4$ and $C_5$ are constants that only depend on $\delta_{2k}$. 
\end{lem}

\begin{proof}[Proof of Theorem \ref{withmsq}]
 Let $x \in \mathbb{R}^n$ be the given signal, $y=\Phi x+\eta$ be the measurement vector (as usual), and $\tilde{y} :=y_{MSQ}$ be the vector obtained from $y$ be performing \textit{MSQ} (with the step size $\delta'$ mentioned above). Then, we have  $\tilde{y}=y+\eta'=\Phi x +\eta+\eta'$ with $\|\eta\|_2 \leq \epsilon \sqrt{m}$ and  $\|\eta'\|_2=\|y-\tilde{y}\|_2 \leq \frac{\delta'}{2} \sqrt{m} $. Moreover, since we apply $\Sigma \Delta$ quantization scheme on the vector $U\tilde{y}$ (to obtain the quantized vector $q$), we can write  $$U\tilde{y}-q=D^ru$$ with $\|u\|_2 \leq C_r \delta \sqrt{m}$ \cite{rongrong}. Thus, $$U \Phi x-q=D^ru+\mu=Hp'$$ where $\mu=U(\eta+\eta')$, $H=[C_r D^r \ \frac{\delta''}{\delta}I] $  and $p'=\begin{aligned}
     \begin{bmatrix}
           \frac{1}{C_r} u \\
           \\
           \frac{\delta}{\delta''} \mu \\

         \end{bmatrix}
  \end{aligned} $. Note that $$\|\mu\|_2 \leq \|\eta\|_2+\|\eta'\|_2 \leq (\epsilon+\frac{\delta'}{2})\sqrt{m}=\delta'' \sqrt{m}$$ where $\delta''$ is defined as $\delta''=\epsilon+\frac{\delta'}{2}$. Hence, $\|p'\|_2^2 \leq \delta^2 m + \delta^2 m=2 \delta^2 m$. \\
 
\noindent Therefore, by defining $w:=D^{-r}(\Phi \hat{x}-\hat{\nu}-q)$ where  $\hat{x}$ and $\hat{\nu}$ are the solutions to minimization problem \eqref{onestagesol2} , we have \begin{equation} \begin{aligned} \|H^{\dagger} U \Phi (x-\hat{x})\|_2 & \leq \|H^{\dagger}(U \Phi x-q)\|_2+\|H^{\dagger}(U \Phi \hat{x}-q)\|_2 \\ & =\|H^{\dagger}Hp'\|_2+\|H^{\dagger}(D^rw+\hat{\nu})\|_2 \\ 
 & \leq \|H^{\dagger} H \|_{op} \|p'\|_2+\|H^{\dagger}H\|_{op} \|p''\|_2 \\
  & \leq \delta \sqrt{2m}+ \delta \sqrt{2m}=2 \delta \sqrt{2m}  \end{aligned} \end{equation}

  \noindent where $ p''= \begin{aligned}
     \begin{bmatrix}
           \frac{1}{C_r} w \\
           \\
           \frac{\delta}{\delta''} \hat{\nu} \\

         \end{bmatrix}
  \end{aligned}$  , and $\|p''\|_2 \leq \delta \sqrt{2m}$ (since $\|w\|_2 \leq C_r \delta \sqrt{m}$ and $\|\hat{\nu}\|_2 \leq \delta'' \sqrt{m}$ by \eqref{onestagesol2}). \\
 
\noindent  On the other hand for every $1 \leq \ell \leq m$, 
 
 $$\begin{aligned}  2 \delta \sqrt{2m} \geq \|H^{\dagger}U\Phi (x-\hat{x})\|_2 & =\|V \Sigma^{-1} U^T U \Phi (x-\hat{x}) \|_2 \\ & =\| \Sigma^{-1} \Phi (x-\hat{x})\|_2 \geq \sigma_{\ell} (H^{\dagger}) \|\Phi_{\ell} (x-\hat{x})\|_2 \end{aligned}$$
 
 \noindent where $\sigma_{\ell}(H^{\dagger})$ is the $\ell th$ singular value of $H^{\dagger}$. Hence, $$\|\Phi_{\ell}(x-\hat{x})\|_2 \leq \frac{2 \delta \sqrt{m}}{\sigma_{\ell}(H^{\dagger})} $$
 
 \noindent and the lower bound for $\sigma_{\ell}(H^{\dagger})$ is given in (22) of \cite{rongrong} (with $\epsilon$ replaced by $\delta''$). Now by Lemma \ref{hatx} , if $k$ and $\ell$ are chosen so that $\frac{1}{\sqrt{\ell}} \Phi_{\ell}$ satisfies RIP with $\delta_{2k} <1/9$, then by using $p=2$ we obtain $$\|x-\hat{x}\|_2 \leq \frac{C_4}{\sqrt{\ell}} \|\Phi_{\ell} (x-\hat{x})\|_2+\frac{2C_5}{\sqrt{k}} \sigma_k (x)$$ where we used the fact that $\|\hat{x}\|_1 \leq \|x\|_1$. Hence, for such $\ell$ and $k$ : $$ \begin{aligned} \|x-\hat{x}\|_2 & \leq \frac{C_4}{\sqrt{\ell}} \|\Phi_{\ell} (x-\hat{x})\|_2+\frac{2C_5}{\sqrt{k}} \sigma_{k} (x)  \\
 & \leq \frac{2 C_4 \delta \sqrt{m}}{\sigma_{\ell}(H^{\dagger}) \sqrt{\ell}}+\frac{2C_5}{\sqrt{k}} \sigma_{k} (x) \end{aligned} $$ Now, we use (22) of \cite{rongrong}, with $\frac{\epsilon}{\delta}$ replaced by $\frac{\delta''}{\delta}=\frac{\epsilon+\delta'/2}{\delta}$ to simplify the bound above. $$\begin{aligned} \|x-\hat{x}\|_2 &  \leq 2 C_4 C_r \delta (3 \pi r)^r (\frac{m}{\ell})^{-r+1/2} +2C_4 \sqrt{\frac{m}{\ell}} (\epsilon+\delta'/2)+\frac{2C_5}{\sqrt{k}} \sigma_k(x) \\ & \leq 2C_4 \delta (3 \pi r)^r (\frac{m}{\ell})^{-r+1/2}+2C_4 \sqrt{\frac{m}{\ell}} \epsilon+\frac{2C_5}{\sqrt{k}} \sigma_k(x) \end{aligned} $$ for values of $m$ satisfying $\ell \leq m \leq m_{max}$. In the last inequality above, we used $\ell \geq 1$, and we assumed $C_r=1/2$, and $\delta'=\frac{\delta (3 \pi r)^r}{m_{max}^r}$. 
 

 \end{proof}

\subsection*{Numerical experiments}

In this section, we verify  the result given in Theorem \ref{withmsq} empirically. In order to do that we repeat the experiment explained in Section  \ref{nmodified}. The only difference is that in this experiment, to obtain the measurement vector $y$, we use the original random partial Fourier matrix $\mathcal{F}^{m \times 200}$ (as opposed to $U \mathcal{F}^{m \times 200}$), then we use the step size $\delta'=\frac{(3 \pi r)^r \delta}{m_{\max}^r}$ to obtain the high-budget quantized vector $y_{MSQ}$ (which will be stored in the buffer and will be discarded later). Next, we find $Uy_{MSQ}$ and quantize it using $r$th order ($r=1,2$) $\Sigma \Delta$ quantization (with the step size $\delta=0.1$)  to obtain the vector $q$. Next, we use \eqref{lm6} to obtain the vector $\hat{x}$ and we find the error in approximation. Similar to what we did in Section \ref{nmodified}, we repeat the experiment for 20 signals, and we take an average for the error in approximation. The graph of errors along with the reference graphs $f(m)=\frac{C}{\sqrt{m}}$ and $g(m)=\frac{D}{m \sqrt{m}}$ are shown in Figure \ref{withextradelta} in log-log scale.


\begin{figure}
    
\begin{center}
    \includegraphics[width=0.35\textwidth]{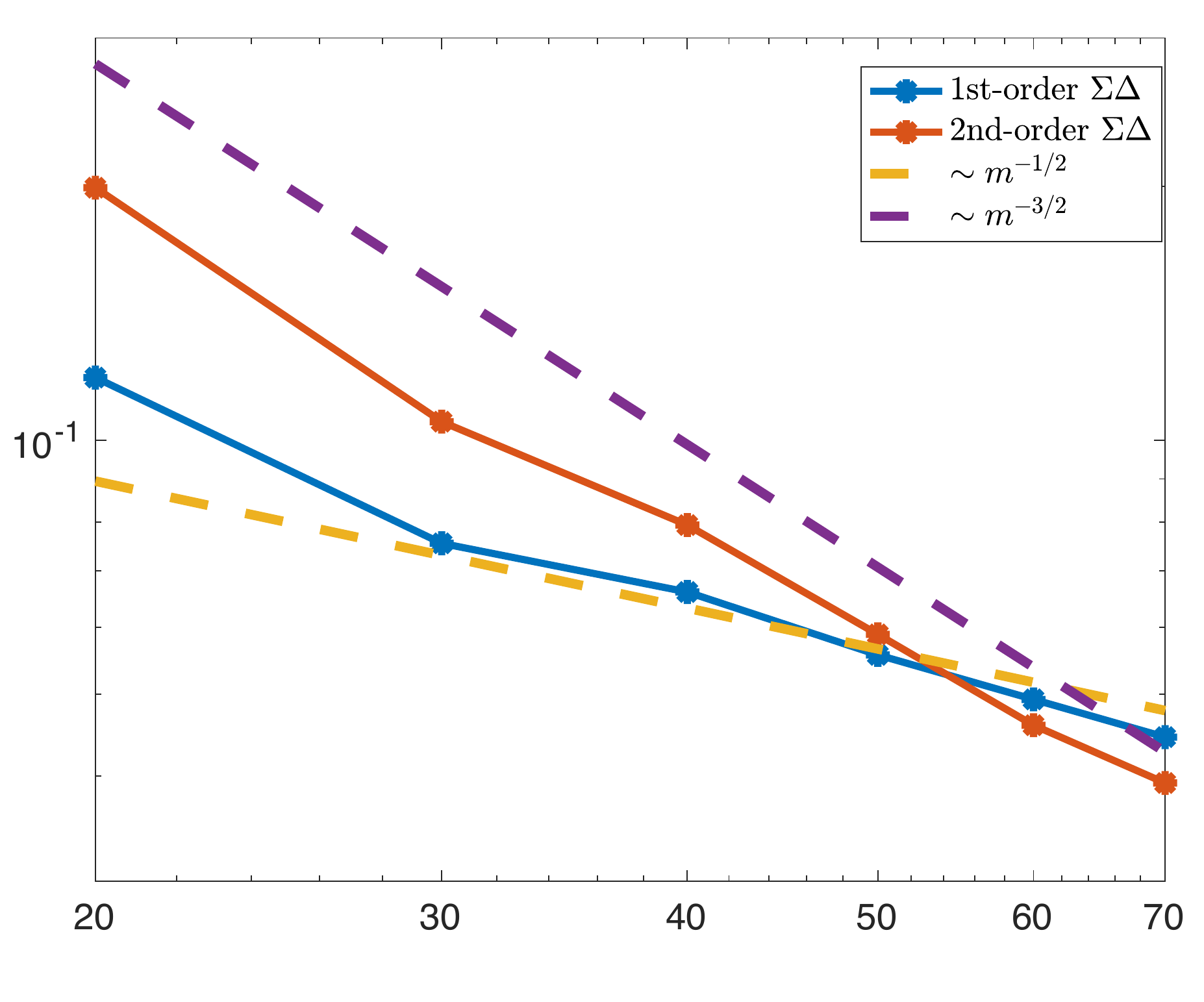}
\end{center}
\caption{ \label{withextradelta} Error in approximation  using  first order and second order $\Sigma \Delta$ quantization with one-stage reconstruction scheme and with an extra MSQ step (before applying  the matrix $U$ on the measurement vector).   }

\end{figure}

\section{One-stage recovery for $\Sigma \Delta$ quantization with deterministic matrices}
\label{chirpcannot}

\textit{Chirp sensing} matrices constitute an important class of deterministic matrices, first introduced by Applebaum et al. \cite{applebaum} in the context of CS. For a prime number $p$ and $\omega:=e^{i\frac{2\pi}{p}}$, the columns of a $p \times p^2$ chirp sensing matrix $\Phi$ are defined via 
\small
\begin{equation} \label{imp} 
\Phi_{rp+m+1}=\left[\omega^{r\cdot{0}^2+m\cdot{0}}, \omega^{r\cdot{1}^2+m\cdot{1}}, \dots, \omega^{r\cdot{(p-1)}^2+m\cdot{(p-1)}} \right]^T
\end{equation}
\normalsize
where $r$ and $m$ range between $0$ and $p-1$. 
As in the case of random measurement matrices, it is natural to ask whether $\Sigma \Delta$ schemes can be used to quantize CS measurements obtained using chirp sensing matrices. 

%

%
%
Motivated by the fact that chirp sensing matrices can be used as CS measurement matrices, we try to use $\Sigma \Delta$ schemes to quantize CS measurements obtained using these matrices. We know that we can do so if they satisfy \textit{(P1)}. However, we observe that \textit{(P1)} does not hold for these matrices.  Consider a $p \times p^2$ chirp sensing matrix $ \Phi$ and let   $T=\{1,2\}$ (hence, we shall consider the first and second columns of this matrix). Note that we prefer the parameter $\ell$ in \eqref{lm6} to be as small as possible in order to minimize the error in approximation, but as we illustrate below the property \textit{(P1)} does not hold even for $\ell=p^{1-\epsilon}$ (for any $\epsilon>0$), and large enough $p$. Set $E=\Phi_{\ell}$ (where as above, $\Phi_{\ell}$ denotes the restriction of $\Phi$ to its first $\ell$ rows). Next, consider the matrix $A=\frac{1}{\ell} E_{T}^* E_T$. Obviously, $A_{11}=A_{22}=1$ and for any given $\epsilon'>0$, $$|A_{12}|=|A_{21}|=\Big| \frac{1}{\ell} \Big(1+e^{\frac{ 2 \pi i}{p}} + e^{\frac{ 4 \pi i}{p}} + e^{\frac{6 \pi i}{p}}+...+e^{\frac{ 2 \ell \pi i}{p}} \Big) \Big| \geq 1- \epsilon' $$ for large enough $p$ since each term in the sum above goes to 1 as $p \to \infty$. The eigenvalues of this matrix satisfy $(1-\lambda)^2 = |A_{12}|^2 $ and so $\lambda_{min} = 1-|A_{12}| \leq \epsilon'$. Hence, $\delta_2=\max\{1-\lambda_{min}, \lambda_{max}-1\} \geq 1-\epsilon'$ for large enough $p$, and therefore $\delta_2<1/9$ can not hold. 

 However, this issue can be resolved if we use a certain submatrix of the chirp sensing matrix by choosing certain values of $m$. Specifically, we define a $p \times p \lfloor \sqrt{p} \rfloor$ matrix $\bar{\Phi}$ as follows. 

 \begin{defi} \label{barphi} Let  $p$ be a prime number, and $\Phi=(\omega^{rx^2+mx})_{x \in \mathbb{Z}_p} $ be a $p \times p^2$ chirp sensing matrix, where the columns are indexed by two parameters $r$ and $m$ in $\mathbb{Z}_p$.  Define $\bar{\Phi}=(\omega^{rx^2+mx})_{x \in \mathbb{Z}_p}$ as a $p \times p \lfloor \sqrt{p} \rfloor$  submatrix of $\Phi$ if the values of $r$ and $m$ are chosen from $\{ 0,1,2.,,,p-1\}$, and $\{ \lfloor \sqrt{p} \rfloor, 2 \lfloor \sqrt{p} \rfloor, ...., (\lfloor \sqrt{p} \rfloor)^2 \}$ respectively. 
 \end{defi}  

We will show that such matrices satisfy \textit{(P1)} and hence, one can perform one-stage $\Sigma \Delta$ quantization using them as measurement matrices.  We will analyze the corresponding approximation error in two scenarios: First, we fix the sparsity level and vary the number of measurements. Next,  we fix the number of measurements and vary the sparsity level.

\subsection{Approximation error as the number of measurements grows}
\label{pgrows}

 In this case, we fix the signal $x$ and we will increase the ambient dimension and the number of measurements while fixing the vector $x$ by embedding $x$ into higher dimensional space. This is because for the class of matrices defined above, to increase the number of measurements $p$, we must also increase the ambient dimension, which is equal to $p \lfloor \sqrt{p} \rfloor$.  As such, we evaluate the error in quantization using one-stage  $\Sigma \Delta$ quantization as the number of measurements $p$ increases. 
 
  First, we prove that the class of matrices defined in Definition \ref{barphi} satisfy the property (\textit{P1}) of order $(k,\ell)$ for appropriate choices of $k$, and $\ell$. 

\begin{Theor} \label{measurementgrows} Consider the $p \times p \lfloor \sqrt{p} \rfloor $  matrix $\bar{\Phi}$ as defined in Definition \ref{barphi}. Then there exists a prime number $p_0$ such that for $p \geq p_0$, the matrix $\Phi$ satisfies the property (P1) of order $(k,\ell)$ for $k \leq \sqrt[4]{p} \log p $ and $\ell =\lfloor p^{3/4} \log^2p \rfloor$.

%

\end{Theor}

To prove this theorem, we will use the following result about an estimate for exponential sums, given by Weyl \cite[p. 41]{montgomery}.  

\begin{Theor}

\label{weylquadratic}

\noindent (a) Suppose that $P(x)=\alpha x^2+ \beta x + \gamma$ where $\alpha$ satisfies $$\Big| \alpha-\frac{a}{q} \Big| \leq \frac{1}{q^2}$$ for some relatively prime integers $a$ and $q$. Then $$\Big| \sum_{n=1}^{N} e(P(n)) \Big| \lesssim \frac{N}{\sqrt{q}}+\sqrt{N \log q}+\sqrt{ q \log q} $$ where $e(x)=e^{2 \pi i x}$ and the notation $f \lesssim g$ means $|f| \leq C g$ for a constant $C$ and for all values of the free variables under consideration. \\

\noindent (b) Suppose that $P(x)=\beta x + \gamma $. Then $$\Big| \sum_{n=1}^{N} e(P(n)) \Big| \leq \frac{1}{2 \|\beta \|}$$ where $\| \beta \|$ is the distance to the nearest integer.  \\
\end{Theor}

\begin{proof}[Proof of Theorem \ref{measurementgrows}] First we define the \textit{incomplete Gauss type} sum $S(r,m,p,\ell)$  for $r$, $m$, $p$, and $\ell$ as given in the theorem via 

$$S(r,m,p,\ell):=e^{\frac{2 \pi i}{p} [r+m]}+e^{\frac{2 \pi i}{p} [4r+2m]}+...+e^{\frac{2 \pi i}{p} [r \ell^2 + m \ell]}$$

\noindent Suppose that  $v_a$ and $v_b$ are two distinct columns of $\frac{1}{\sqrt{\ell}} \Phi_{\ell}$ corresponding to the values of $r_1 , m_1$ and $r_2 , m_2$ (i.e., $a=r_1p+m_1+1$ and $b=r_2p+m_2+1$). Then \begin{equation} \label{vavb}  | \langle v_a , v_b \rangle | = \frac{1}{\ell} |S(r_2-r_1,m_2-m_1,p,\ell)| \end{equation}

\noindent To bound the RHS in above, we need to consider two cases.

 \textbf{Case 1.} If $r_1 \ne r_2$, we bound $ | \langle v_a , v_b \rangle |$ by setting $\ell=\lfloor  p^{3/4} \log^2 p \rfloor$. For this purpose, we apply part (a) of Theorem \ref{weylquadratic} mentioned above with  $\alpha=\frac{r_2-r_1}{p}$ which is of the form $\frac{a}{q}$ for relatively prime integers $a=r_2-r_1$ and $q=p$. Hence, by using part (a) of Theorem \ref{weylquadratic}, and  using the fact that $\frac{p^{3/4} \log ^2 p}{2} \leq p^{3/4} \log^2 p -1 \leq \ell \leq  p^{3/4} \log^2 p $ for $p \geq 3$, we obtain  \begin{equation}  \label{inner} \begin{aligned} | \langle v_a , v_b \rangle | & \lesssim  \frac{1}{p^{3/4} \log^2 p} |S(r_2-r_1,m_2-m_1,p,\ell)| \\&  \lesssim  \frac{1}{p^{3/4} \log^2 p} \Big( p^{1/4} \log^2p+p^{3/8} \log^{3/2} p+\sqrt{p \log p} \Big) \lesssim  \frac{1}{p^{1/4} \log^{3/2} p} \end{aligned} \end{equation} 

\noindent \textbf{Case 2.}  If $r_1 = r_2$, we set $\ell= \lfloor p^{3/4} \log^2 p \rfloor $ and we use the fact that the set of possible values of $m$ are $\{\lfloor \sqrt{p} \rfloor , 2\lfloor \sqrt{p} \rfloor ,3 \lfloor \sqrt{p} \rfloor,....,(\lfloor \sqrt{p} \rfloor)^2 \} $, and we use part (b) of Theorem \ref{weylquadratic}. Note that in our problem $\beta=\frac{m_2-m_1}{p}$, and $ -1 \leq \beta \leq 1$. Accordingly, to evaluate $\| \beta\|$, we evaluate $\min \{ |\beta|, 1-|\beta| \}$. Next, $$| \beta| \geq \frac{\lfloor \sqrt{p} \rfloor}{p} \geq \frac{\sqrt{p}-1}{p} \geq \frac{0.1 \sqrt{p}}{p}=\frac{0.1}{\sqrt{p}}$$ for any prime number $p$. Also, 

$$|\beta| \leq \frac{(\lfloor \sqrt{p} \rfloor -1) \lfloor \sqrt{p} \rfloor }{p} \leq \frac{(\sqrt{p}-1)\sqrt{p}}{p}=1-\frac{1}{\sqrt{p}}$$
\noindent which implies $1-|\beta| \geq \frac{1}{\sqrt{p}}$. Therefore, in any case,  $\| \beta \| \geq \frac{0.1}{\sqrt{p}} $,  and since we set the value of $\ell$ to be $\ell=\lfloor p^{3/4} \log^2 p \rfloor $, we will have \begin{equation} \label{inner2} |\langle v_a , v_b \rangle | \lesssim \frac{1}{p^{3/4} \log^2 p}  \sqrt{p}  \lesssim \frac{ 1}{p^{1/4} \log^2 p}.  \end{equation}   Combining the equations \eqref{inner} and \eqref{inner2} we obtain  $$|\langle v_a , v_b \rangle | \lesssim \frac{1}{p^{1/4}\log^{3/2} p}$$ Since the columns of the matrix $\Phi_{\ell}$ have unit norm, we can conclude that the coherence of this matrix, $\mu$, satisfies $\mu \lesssim \frac{1}{p^{1/4}\log^{3/2} p}$.  Therefore, there exists a prime number $p_0$ such that for $ p \geq p_0$ : $$ \delta_k \leq k \mu \lesssim  p^{1/4} \log p \frac{1}{p^{1/4}\log^{3/2} p}=\frac{1}{\sqrt{ \log p}} < \frac{1}{9} $$
\end{proof}

\noindent Next, we prove  the following  corollary which shows that we can indeed use $\bar{\Phi}$ along with a $r$th-order $\Sigma \Delta$ quantizer. 
%


\begin{cor}  \label{cor11}
Let $x \in \Sigma_{k}^n$, let $p_0$ be as defined in Theorem~\ref{measurementgrows}, and suppose that $p_1>p_0$ is a prime number such that $k \leq \sqrt[4]{p_1} \log p_1$. Then, for any $p \geq p_1$, $x$ can be approximated by $\hat{x}$, the solution to \eqref{onestagesol}, if 
\begin{enumerate}
\item the measurement matrix is $U \bar{\Phi}$, where $\bar{\Phi}$ is the $ p \times p \lfloor \sqrt{p} \rfloor $ matrix defined as in Definition \ref{barphi},  and 
\item $q$ is obtained by quantizing $U \bar{\Phi}$ using an $r$th order $\Sigma \Delta$ scheme. 
\end{enumerate}
In the noise-free case,  as we increase the number of measurements $p$, the approximation error satisfies  
\begin{equation} \label{kfixed} \|x-\hat{x} \|_2 \leq  C (3 \pi r)^r \delta (\log p)^{2r-1}   p^{-\frac{1}{4}(r-\frac{1}{2})}        \end{equation} 
where $C$ is a constant that does not depend on $r, p_0$, $p$, and $p_1$. 
 \end{cor}

 \begin{proof} 
 

  Set $\ell= \lfloor p^{3/4} \log^2 p \rfloor $. Then, since $p>p_1$, we have $k \leq \sqrt[4]{p} \log p$. Thus,   by Theorem \ref{measurementgrows},  the $p \times p \lfloor \sqrt{p} \rfloor $ matrix $\bar{\Phi}$ satisfies \textit{(P1)} of order $(k,\ell)$, and hence  the vector $x$ can be approximated by $\hat{x}$. Moreover, by Theorem \ref{recoveryp1}, as $p$ increases in the noise-free case, the error in approximation satisfies  
 
 \begin{equation} \label{kfixed} \|x-\hat{x} \|_2 \leq  C (3 \pi r)^r \delta (\log p)^{2r-1}   p^{-\frac{1}{4}(r-\frac{1}{2})}        \end{equation} 
where $C$ is a constant that does not depend on $r, p_0$, $p$, and $p_1$.

%
%

 \end{proof}


 Note that the error decay rate $O( p^{-\frac{1}{4}(r-\frac{1}{2})} )$  (up to a factor logarithmic in $p$)  given in Corollary~\ref{cor11} is inferior to $O(p^{-(r-\frac{1}{2})})$ which we obtain with random matrices (with $m=p$ measurements). This behaviour is due to the fact that the both dimensions of $\bar{\Phi}$ increase as we increase $p$. One way to circumvent this issue is to restrict the maximum number of measurements to some  $p_{\rm{max}}$. In the following theorem, we will prove that under such circumstances, the approximation error behaves like $p^{-(r-\frac{1}{2})}$, similar to the case with random matrices. 
 
 
  \begin{Theor} \label{alphatheor}
 
  Fix $ \alpha , \beta >0$, with $\alpha+\beta/2 < 1/2$. Let $x \in \Sigma_{k}^n$, and assume that  $p_0$ be as defined in Theorem~\ref{measurementgrows}. Suppose that  $p_1>p_0$ is a prime number such that $k \leq p_1^{  \alpha}$. Then, for any $p_1 \leq p \leq p_{max}$, where $p_{max}=\mathcal{O}(p_1^{1+\beta})$, the signal  $x$ can be approximated by $\hat{x}$, the solution to \eqref{onestagesol}, if

  \begin{enumerate}
 \item the measurement matrix is $U \bar{\Phi}$,  where $\bar{\Phi}$ is the $ p \times p \lfloor \sqrt{p} \rfloor $ matrix defined as in Definition \ref{barphi},  and 
 \item $q$ is obtained by quantizing $U \Phi x$ using an $r$th order $\Sigma \Delta$ scheme. 
 \end{enumerate}
 In the noise-free case,  as we increase the number of measurements $p$, the approximation error satisfies  
 \begin{equation} \label{kfixed22} \|x-\hat{x} \|_2 \leq  D  \delta  p^{-(r-1/2)}        \end{equation} 
 where $D$ is a constant that depends on $p_1$, and order $r$, but does not depend on $p_0$ or $p$.

  \end{Theor} 
 
  \begin{proof} Set $\ell=\lfloor p_1^{1/2+\alpha+\beta/2} \log ^2 p_1 \rfloor $.  Then, by using Theorem \ref{weylquadratic}, and  similar to the argument given for the proof of Theorem \ref{measurementgrows}, we conclude that the coherence of $\tilde{\Phi}$ satisfies $$\mu \lesssim \frac{1}{\ell} \Big( \frac{\ell}{\sqrt{p}}+\sqrt{ \ell \log p}+\sqrt{p \log p} +\sqrt{p} \Big) \lesssim \frac{1}{\ell} \sqrt{p \log p} \lesssim \frac{p_1^{1/2+\beta/2} \sqrt{\log p_1}}{p_1^{1/2+\alpha+\beta/2} \log^2 p_1} $$
  
  \noindent Hence, the RIP constant of $\frac{1}{\sqrt{\ell}} \bar{\Phi}$     satisfies $$\delta_k <k \mu \lesssim \frac{p_1^{1/2+\beta/2+\alpha}\sqrt{\log p}}{p_1^{1/2+\alpha+\beta/2} \log^2 p_1}=\frac{1}{\log^{3/2} p_1} < 1/9$$
  
  \noindent where we used  the fact that $\frac{1}{\log^{3/2} p_0 }< 1/9$, and $p_1 \geq p_0$. This means that $\bar{\Phi}$ satisfies the property \textit{(P1)} of order $(k,\ell)$, and hence the vector $x$ can be reconstructed using the solution of \eqref{onestagesol} if $U \bar{\Phi}$ is used as the measurement matrix. Also, by using $\ell=\lfloor p_1^{1/2+\alpha+\beta/2} \log ^2 p_1 \rfloor $, and $m=p$ in \eqref{onestagesol} we obtain the bound on the error in approximation \eqref{kfixed22} as desired. \end{proof}  
  
%
%
 
  As an example of Theorem above, we can set $k=4$, $\alpha=0.34$, and $\beta=0.3$. Then $\alpha+\beta/2=0.49<1/2$, and we must choose $p_1$ such that $k=4<p_1^{0.34}$. We can observe that $p_1=61$ satisfies this inequality. Hence, if the number of measurements satisfies $61 \leq p \leq 61^{1.3}$, then the guarantee on the error bound \eqref{kfixed22} will hold. 
  
%


 Combining Theorem \ref{alphatheor} and Corollary \ref{cor11}, we obtain the following result.

 \begin{cor} \label{alphacor}
 
  Fix $ \alpha , \beta >0$, with $\alpha+\beta/2 < 1/2$. Let $x \in \Sigma_{k}^n$, and assume that  $p_0$ be as defined in Theorem~\ref{measurementgrows}. Suppose that  $p_1>p_0$ is a prime number such that $k \leq p_1^{  \alpha}$. Then, for any $p \geq p_1$,  the signal  $x$ can be approximated by $\hat{x}$, the solution to \eqref{onestagesol}, if

  \begin{enumerate}
 \item the measurement matrix is $U \bar{\Phi}$,  where $\bar{\Phi}$ is the $ p \times p \lfloor \sqrt{p} \rfloor $ matrix defined as in Definition \ref{barphi},  and 
 \item $q$ is obtained by quantizing $U \Phi x$ using an $r$th order $\Sigma \Delta$ scheme. 
 \end{enumerate}
 In the noise-free case,  as we increase the number of measurements $p$, the approximation error satisfies  
 \begin{equation} \label{kfixed22} \|x-\hat{x} \|_2 \leq  D  \delta  p^{-(r-1/2)}        \end{equation} 
 if $p \leq p_1^{1+\beta}$, and    \begin{equation} \label{kfixed} \|x-\hat{x} \|_2 \leq  C \delta (\log p)^{2r-1}   p^{-\frac{1}{4}(r-\frac{1}{2})}        \end{equation} 
if $p>p_1^{1+\beta}$.

  \end{cor}

\subsection{Approximation error as the sparsity level  varies}
\label{kvaries}

In the previous section, we saw that if we use an appropriate measurement matrix, and an appropriate approximation scheme, then as we increase the number of measurements, the error in approximation decreases. Our objective in this section is to fix the number of measurement (which also fixes the ambient dimension) and reduce the sparsity level $k$. We expect to observe a similar behaviour to what we observed above,  and see a decay in error in approximation. \\

\begin{Theor} \label{p1def} Consider the CS matrix  $\bar{\Phi}$ as defined in Definition \ref{barphi}. There exists a prime number $p_1$ such that for a fixed number of measurements $p$, with $p \geq p_1$,  $1 \leq k \leq \frac{ \sqrt{p}}{\log p} $ and $\ell =\lfloor k \sqrt{p} \log p \rfloor $, the matrix $\bar{\Phi}$ satisfies the property (P1) of order $(k,\ell)$.  \end{Theor}

\begin{proof} By doing a similar calculation to the one given in the proof of Theorem \ref{measurementgrows}, and using $\ell=\lfloor k \sqrt{p} \log p \rfloor$, the equations \eqref{inner} and \eqref{inner2} will be replaced by $$  |\langle v_a , v_b \rangle |  \lesssim \frac{1}{k \sqrt{p} \log p} \Big(  k \log p+ \sqrt{k} \sqrt[4]{p} \log p +\sqrt{p \log p}\Big) $$ and $$|\langle v_a, v_b \rangle | \lesssim \frac{\sqrt{p}}{ k \sqrt{p} \log p} $$ respectively. This implies  $$ |\langle v_a, v_b \rangle| \lesssim \max\{\frac{1}{\sqrt{p}}, \frac{1}{\sqrt{k} \sqrt[4]{p}}, \frac{1}{k \sqrt{ \log p}}, \frac{1}{k \log p} \}   $$Therefore there exists a prime number $p_1$ such that for $p \geq p_1$, the \textit{RIP} constant of the matrix $\Phi_{\ell}$ satisfies $$\delta_k < k \mu \lesssim \max\{ \frac{k}{\sqrt{p}}, \frac{\sqrt{k}}{\sqrt[4]{p}}, \frac{1}{\sqrt{\log p}}, \frac{1}{\log p} \} \lesssim \frac{1}{\sqrt{\log p}} < 1/9 $$ where we used the assumption on the sparsity level $k \leq \frac{\sqrt{p}}{\log p}$. \end{proof}

 \noindent Similar to what we observed in Section \ref{pgrows}, we state a corollary regarding the bound on the error term when the matrix $U\bar{\Phi}$ is used as the measurement matrix, and one-stage recovery scheme is used to reconstruct $x$. To match this corollary with the similar results, where we had a decreasing function for the error term, we consider the error as a function of $k'=1/k$. Note that we expect the error term to decay as we \textit{decrease} the value of $k$, i.e., as we \textit{increase} the value of $k'$.

 \begin{cor}

   There exists a prime number $p_0$ such that for a fixed prime number $p$ with $p \geq p_0$, any $k$-sparse signal $x$ can be approximated with  the vector $\hat{x}$, the solution to \eqref{onestagesol}, provided that the following holds. 
   \begin{enumerate}
   
 \item The sparsity level satisfies $k \leq \lfloor k_{\max}:=\frac{\sqrt{p}}{\log p} \rfloor $.  

\item The measurement matrix is $U \bar{\Phi}$, with $\bar{\Phi}$ defined as in Definition  \ref{barphi}; 

\item $q$ is obtained by quantizing $U \bar{\Phi} x$ using an $r$th order $\Sigma \Delta$ scheme -- as in \eqref{sigmadefi}.

\end{enumerate} 
\noindent The error in approximation satisfies


  \begin{equation} \label{pfixed} \|x-\hat{x}\|_2 \leq \Bigg( 2 C_r C_4  (3 \pi r)^r(\frac{\sqrt{p}}{\log p})^{-r+1/2} \Bigg) (k')^{-r+1/2} \end{equation} assuming that no noise is present. Here, $k'=1/k$, and the constant $C_4$ only depends on RIP constant of $\Phi$.

\end{cor}
\label{aboutksparse}
 
\begin{proof} Let $x$ be a $k$-sparse signal, and $\hat{x}$ be the approximation vector. Also,  let $p_0$ be the prime number given by Theorem \ref{measurementgrows}. Replace the value of $m=p$ and $\ell=\lfloor k \sqrt{p} \log p \rfloor \leq k \sqrt{p} \log p $ into \eqref{lm6} and use the fact that $\sigma_k(x)=0$ for a $k$-sparse signal to conclude

$$\|x-\hat{x} \|_2 \leq \Bigg( 2 C_r C_4  (3 \pi r)^r(\frac{\sqrt{p}}{\log p})^{-r+1/2} \Bigg) (k)^{r-1/2} $$ as desired. 

 \end{proof}

\subsection*{Numerical experiments}

In this section, we verify the results we obtained in Sections \ref{pgrows} and \ref{kvaries}. We run two numerical experiments. In the first experiment, we consider prime numbers $p=$ 61, 137, 223, 307, 397, 487, 593, 677, 787, and  for each prime $p$, we draw 20 signals, each of which is a 4-sparse signal with a random support chosen from the set $\{1,2,\cdots, 61 \big\lfloor \sqrt{61} \big\rfloor\} $, and whose entries are chosen independently from a standard Gaussian distribution. In other words, the actual ambient dimension of signals that are considered is $61 \big\lfloor \sqrt{61} \big\rfloor=427$. For each such signal, we compute the CS measurements $y=U\bar{\Phi}$ which we subsequently quantize using a stable $r$th-order $\Sigma \Delta$ scheme to obtain $q$ with $r=1$ or $r=2$. Next, we reconstruct an approximation $\hat{x}$ of $x$ using  \eqref{onestagesol} where we set $\Phi=U\bar{\Phi}$, $\delta=0.1$, $r=1, 2$, and $\epsilon=0$. Finally, for each $p$, we compute the average $\|x-\hat{x}\|_2$.  We plot the average error as a function of $p$ in log-log scale in  Figure~\ref{componestage5}.  As mentioned  in Section \ref{pgrows}, for 4-sparse signals, we expect the bound on the error in approximation to behave like $p^{-(r-1/2)}$ at least  for $61\leq p \leq 61^{1.3}$. Figure~\ref{componestage5} confirms this fact and shows the $p^{-(r-1/2)}$ behaviour even for $p$ values beyond this range.  

%
 In the second experiment, we fixed the number of measurements to be $p=541$, and we considered $k$-sparse signals with $3 \leq k \leq 15$. Then for each $k'=\frac{1}{k}$, we consider 50 signals which are $k$-sparse and have a random support $T \subseteq \{1,2,...,1400\}$ and have entries chosen independently from the standard Gaussian distribution. For each of these signals, the reconstruction vector $\hat{x}$ is obtained from \eqref{onestagesol} with $r=1$ or $r=2$.   We average  over all the errors for each value of $k$, and we plot the graph of average errors as well as the upper bounds on the error obtained in Section \ref{kvaries} in log-log scale in Figure \ref{componestage6}.

\begin{figure}
    
\begin{center}
    \includegraphics[width=0.35\textwidth]{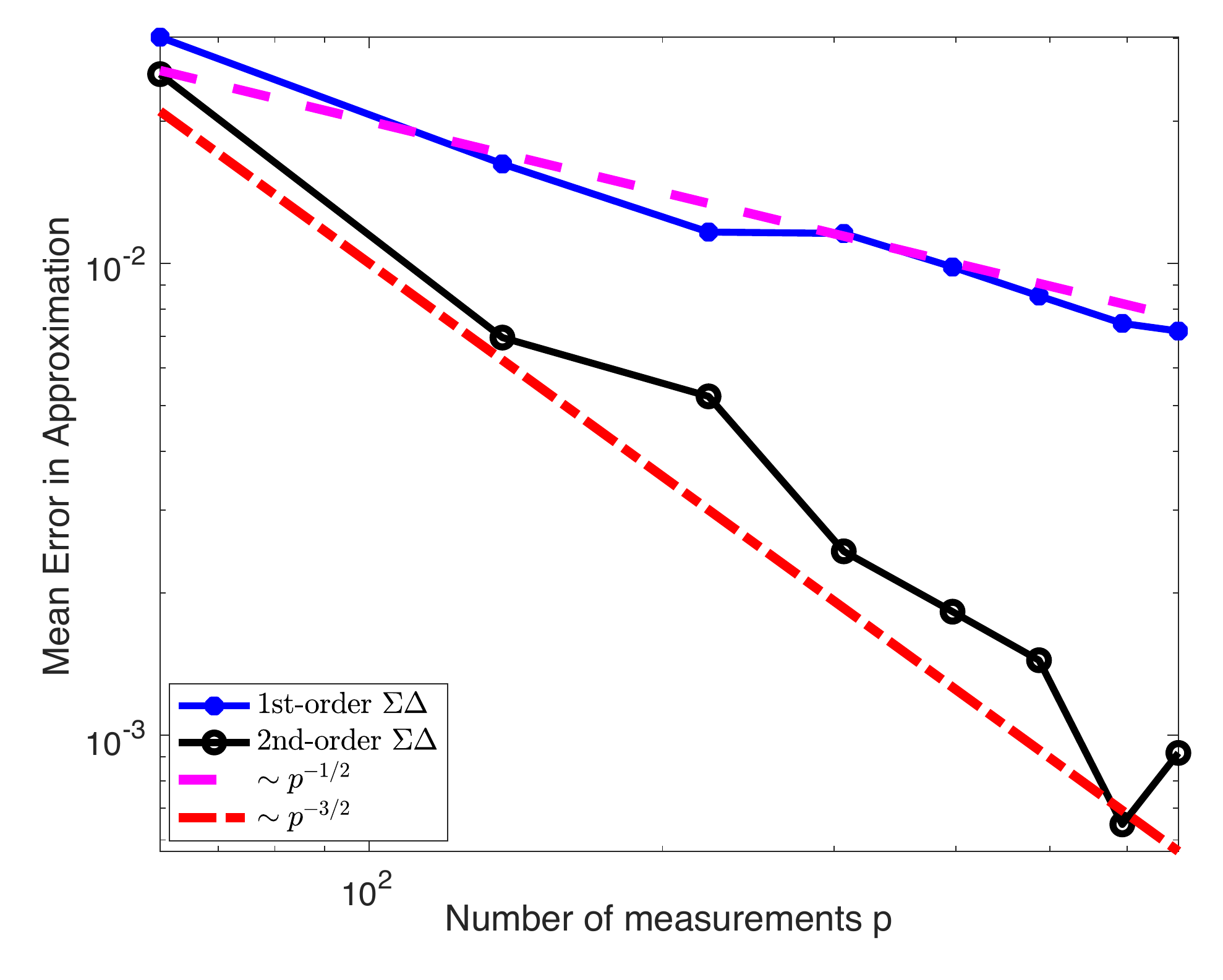}
\end{center}
\caption{ \label{componestage5} Error in approximation using  first order and second order $\Sigma \Delta$ quantization with one-stage reconstruction scheme for a 4-sparse signal and the comparison with the graphs of $f(p)=\frac{C}{\sqrt{p}}$ and $g(p)=\frac{D}{ p^{3/2}}$ (each one shifted properly to match the original graphs) in log-log scale.  }

\end{figure}

\begin{figure}
\begin{center}
    \includegraphics[width=0.35\textwidth ]{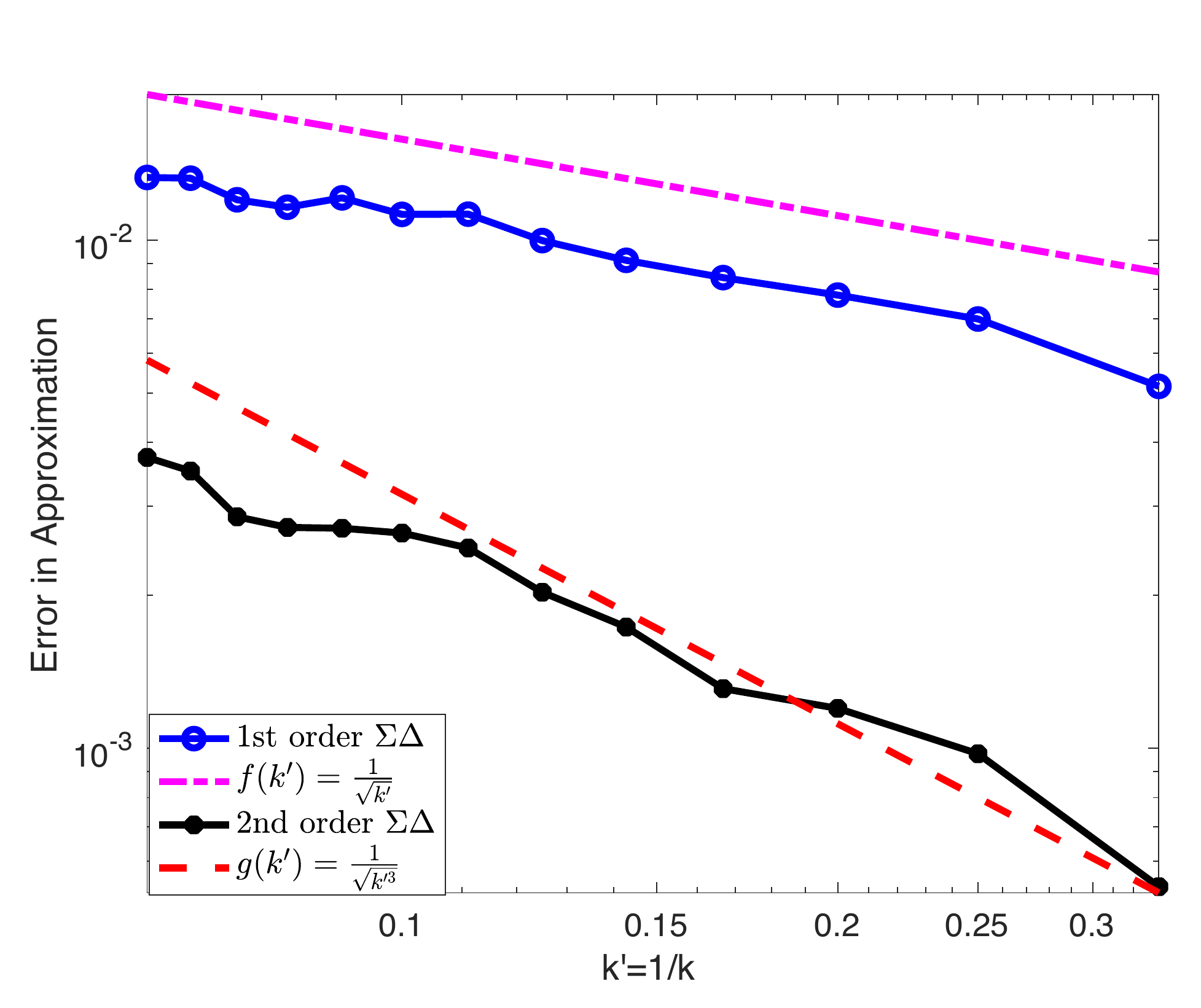}

\caption{\label{componestage6}Error in approximation using first order and second order $\Sigma \Delta$ quantization with one-stage reconstruction scheme with fixed number of measurements ($p=541$) and the comparison with the graphs of $f(k')=\frac{1}{\sqrt{k'}}$ and $g(k')=\frac{1}{\sqrt{k'^3}}$. }
\end{center}

\end{figure}

\section{Further encoding of $\Sigma\Delta$-quantized compressive measurements}

In one-stage recovery of $\Sigma \Delta$ quantized measurements, we start with a measurement vector $y$ and since we have to store/transmit data we quantize this vector using an alphabet $\mathcal{A}$ to obtain a quantized vector $q \in \mathcal{A}^m$. To encode $q$, we need $\log_{2} |\mathcal{A}|^m=m \log_{2} |\mathcal{A}| $ bits. In \cite{encoding}, Saab et al. proposed a method to encode using much less number of bits without affecting the error in reconstruction significantly. In the following, we give a brief review about their result. 

 In a nutshell, they reduce the dimension of $q$ to encode using less number of bits. In particular, suppose that $L \leq m$, and consider the encoder $\mathcal{E}: \mathcal{A}^m \to \mathcal{C}$ defined as $\mathcal{E}(q)=BD^{-r} q$. where $B$ is an $L \times m$ Bernoulli matrix with i.i.d. equiprobable entries. 

First, we find how many bits we are saving by using this encoder. We consider the alphabet $\mathcal{A}_{\delta}^K:=\{-K \delta,...,-\delta,\delta,...,K \delta\}$. Since $\|D^{-r}\|_{\infty} \leq m^r$, and $\|B\| \leq m$ \cite{encoding}, we obtain $\|BD^{-r} q \|_{\infty} \leq m^{r+1}\|q\|_{\infty} \leq m^{r+1} K \delta$. Thus, for each entry of $\mathcal{E}$ we need an alphabet of the form $$\mathcal{A}'=\mathcal{A}_{\delta}^{Km^{r+1}}$$ There are $L$ such entries, so in total we should use $$L\log_{2} |\mathcal{A}'|=L (r+1) \log_2 m + L \log_2 2K$$ bits to represent $\mathcal{E}(q)$. Thus, by enlarging the size of alphabet and reducing the dimension, Saab et al. \cite{encoding} reduced the number of bits because the size of alphabet appears only as logarithmic factor. 

Now, the goal is to find an algorithm to reconstruct $x$ with the vector $\hat{x}$ given the encoded vector $\tilde{q}=\mathcal{E}(q)=BD^{-r}q$, and with $\|x-\hat{x}\|_2$ to be as small as possible. This algorithm is given in \cite{encoding} as follows. \begin{equation} \label{encodingalgorithm} \begin{aligned} (\hat{x}, \hat{u}, \hat{e})  =\mbox{argmin} \|\tilde{x}\|_1 & \mbox{  \ subject to \  } BD^{-r} (\Phi \tilde{x}+\tilde{e}) -B \tilde{u} =BD^{-r} q   \\ & \mbox{ \ and } \|B \tilde{u} \|_2 \leq 3Cm  \mbox{ \ and \ }  \|\tilde{e} \| \leq \sqrt{m} \epsilon \end{aligned} \end{equation} 


Next, we prove that this algorithm can be applied using the measurement matrix defined in Definition \ref{barphi}. In order to do so, first we choose a Bernoulli matrix $B$ of the size $L \times p$ with $L=\lfloor p^{5/8} \log^2 p \rfloor$ and consider the $p \times p$ matrix $D^{-r}$. Then, write the singular value decomposition of $BD^{-r}$ in the form $BD^{-r}=TSR^{T}$. Using this notation, we prove the following theorem. 

\begin{Theor}

Consider a $k$-sparse signal $x \in \mathbb{R}^n$, with $k \leq \lfloor \sqrt[8]{p} \log p \rfloor$.  Suppose that we use the matrix $R \bar{\Phi} $ as the measurement matrix, where $R$ is as above, and $\bar{\Phi}$ is the matrix given in Definition \ref{barphi}, to find the measurement vector $y$. Then, we use the $r$th order $\Sigma \Delta$ quantization to obtain the quantized vector $q$. Next, find the reconstruction vector $\hat{x}$ via \eqref{encodingalgorithm}. The error in reconstruction satisfies  $$\|x-\hat{x} \| \leq C_1 \Big( \frac{ \log p}{\sqrt{p}} )^{r/2-3/4} +C_2\sqrt{\frac{\sqrt{p}}{\log p}} \epsilon +C_3\frac{\sigma_k(x)_1}{\sqrt{k} }   $$ with probability at least $  1-C_5 e^{-c_6 p^{11/16} \log p}    $ for some constants $C_1,C_2,C_3,C_5$, and $c_6$. 
\end{Theor}

\noindent Note that if we want to have decreasing bound (as a function of $p$) for the error in approximation in the noise-free case, we need to have $r/2-3/4>0$. This means we must have $ r \geq 2$.

\begin{proof} First, let $L=\lfloor p^{5/8} \log^2 p \rfloor$, and we verify that  $\frac{1}{\sqrt{L}} \bar{\Phi}_{L} $ satisfies the RIP with  $\delta_{2k} < 1/9$, if $k \leq \sqrt[4]{p} \log p$. To that end,  we use \eqref{vavb} along with Theorem \ref{weylquadratic}, to conclude that $$ \begin{aligned} |\langle v_a,v_b \rangle| & \lesssim \frac{1}{p^{5/8}  \log^2 p} \Big(  p^{1/8} \log^2 p + p^{5/16} \log^{3/2} p+ \sqrt{p \log p}  \Big) \\ &  \lesssim \frac{\sqrt{p \log p} }{p^{5/8} \log^2 p}=\frac{1}{p^{1/8}\log^{3/2} p } \end{aligned}  $$

\noindent Hence, the coherence of $\frac{1}{\sqrt{L}} \bar{\Phi}_{L} $ satisfies $ \mu \lesssim \frac{1}{p^{1/8} \log^{3/2} p}$, and this implies that the RIP constant satisfies $$\delta_k < k \mu \leq \frac{p^{1/8} \log p}{p^{1/8} \log^{3/2} p }<1/9$$ for large enough $p$. 

 Similar to the what mentioned in the proof of Theorem \ref{withmsq}, if we use $p=2$ in Proposition \ref{hatx}, and the value of $L$ as stated above, since $\frac{1}{\sqrt{L}} \bar{\Phi}_L$ satisfies the RIP with $\delta_{2k}<1/9$, we can conclude that  \begin{equation} \label{c1} \|x-\hat{x}\|_2 \leq \frac{d_1}{\sqrt{L}} \| \bar{\Phi}_L (x-\hat{x}) \|_2 +d_2\frac{\sigma_k(x)_1}{\sqrt{k}} \end{equation} 

\noindent for some constants $d_1$ and $d_2$. Next,  we find an upper bound for $\|\frac{1}{\sqrt{L}} \bar{\Phi}_L (x-\hat{x}) \|_2$. To do that, we consider the set $\mathcal{E}=\mathcal{E}_1 \cap \mathcal{E}_2$ where $$\mathcal{E}_1:= \{ B \in Bern(L,m): \sigma_L(BD^{-r}) \geq \Big( \frac{m}{L} \Big)^{r/2-1/4} \sqrt{m} \}, $$ and $$\mathcal{E}_2:= \{ B \in Bern(L,m) : \|B\|_{ \ell^2 \to \ell^2} \leq \sqrt{L} +2 \sqrt{m} \}$$ It is shown in \cite{encoding} that \begin{equation} \label{probepsilon} P(\mathcal{E}) \geq 1- 2e^{-c_1 \sqrt{mL}}-\beta  e^{ -c_2 L}, \end{equation} for some constants $\beta, c_1$, and $c_2$. It is also shown that for any $B \in \mathcal{E}$, if we decompose $BD^{-r}$ in the form $BD^{-r}=TSR^T$, and if we set $\tilde{\Phi}=R^T \Phi$ (here, $\Phi$ is the measurement matrix, and in our case, $\Phi=R \bar{\Phi}$ with $\bar{\Phi}$ as given in Definition \ref{barphi}, and so $\tilde{\Phi}=R^T(R\bar{\Phi})=\bar{\Phi}$), then we have  

 $$\frac{1}{\sqrt{L}} \| \tilde{\Phi}_L (x-\hat{x}) \|_2 \leq 6C \Big( \frac{L}{m} \Big) ^{r/2-3/4} + 2\sqrt{\frac{m}{L}} \epsilon $$ for a constant $C$.  Hence, using the value of $L$ as given above, we obtain  
 
%

\begin{equation} \label{c2} \|\frac{1}{\sqrt{L}} \bar{\Phi}_L (x-\hat{x}) \|_2    \leq  6C \Big( \frac{ \log^2 p}{\sqrt[8]{p}} \Big)^{r/2-3/4} +2\sqrt{\frac{\sqrt[8]{p}}{\log^2 p}} \epsilon \end{equation}

\noindent  Accordingly, by combining \eqref{c1} and \eqref{c2}, we obtain $$\|x-\hat{x}\|_2 \leq C_1 \Big( \frac{ \log^2 p}{\sqrt[8]{p}} )^{r/2-3/4} +C_2 \sqrt{\frac{\sqrt[8]{p}}{\log^2 p}} \epsilon +C_3 \frac{\sigma_k(x)_1}{\sqrt{k}}$$

\noindent Noting that $m=\lceil p^{3/4} \rceil $, and $L=\lceil p^{5/8} \log^2 p \rceil $, we conclude that $e^{-c_1 \sqrt{mL}} \lesssim e^{-c_2 L} $, which implies that $1-e^{-c_2L} \lesssim 1-e^{-c_1\sqrt{mL}}$.  Therefore, by \eqref{probepsilon}, the inequality above holds with probability at least $1-C_5 e^{-c_6 \sqrt{mL}}$, i.e., $1-C_5 e^{-c_6 p^{11/16} \log p }$,  for some constants $C_1, C_2, C_3, C_5$, and $c_6$. 

\end{proof}

\noindent

\noindent Lastly, the following result holds regarding the worst case reconstruction error, i.e., the distortion $\mathcal{D}$ as defined in (1) of \cite{encoding}. The derivation is similar to derivation of (iii) in Corollary 14 from Theorem 12 in \cite{encoding} and is omitted here.


\begin{cor}
There exist  constants $C_0, C_2$ such that in the noise-free case, and for $k :=\lfloor C_0 p^{5/8} \log p \rfloor $, the distortion rate $\mathcal{D}$ in the case of. $k$-sparse signals satisfies $$\mathcal{D} \lesssim 2^{-C_2 \frac{\mathcal{R}}{k \log p}}$$

\noindent where $\mathcal{R}$ is the bit rate defined as $\mathcal{R} := \log |\mathcal{C}|$.

\end{cor}

\section{Conclusion}

In today's digital world, quantizing the measurement vector is a crucial step in the sampling process, which was mostly ignored in early literature of CS. One known efficient method of quantization in CS is a method called $r$th-order $\Sigma \Delta$ quantization, which was accompanied with a one-stage reconstruction method. This method was shown to be robust respect to noise and stable respect to compressible signals, but came with one caveat: it was applied only for the class of sub-Gaussian matrices. In this paper, we proposed two novel approaches to generalize this method to random restrictions of bounded orthonormal systems, such as random restrictions of DFT matrices (which are of high importance due to the applications in MRI). We also generalized this method to certain class of deterministic measurement matrices, namely, certain submatrices of chirp sensing matrices. For each of these cases, we provided numerical experiments confirming the bounds derived for the errors in approximation.

\bibliographystyle{plain}
\bibliography{one_stage_01}



%
%


\end{document}